\documentclass[a4paper]{article}

\usepackage[T1]{fontenc}
\usepackage[utf8]{inputenc}
\usepackage{fullpage, hyperref}

\usepackage{graphicx}
\usepackage{amssymb,amsmath}
\usepackage{amsthm}
\usepackage{authblk}
\usepackage{array, multirow}
\usepackage{chngpage}
\usepackage{booktabs}
\usepackage{color}	
\usepackage{enumerate} 
\usepackage[lofdepth,lotdepth]{subfig}
\usepackage{caption}

\newcommand{\itemshort}{
	\setlength{\itemsep}{0mm}
	\setlength{\parskip}{0mm}
	\setlength{\topsep}{-2mm}
}

\newcommand{\Match}{Ask}
\newcommand{\M}{\mathcal{M}}
\newcommand{\MM}{\mathcal{M}^+ }
\newcommand{\algoName}{\textsc{MaxMatch} }
\newcommand{\Cstart}{D_\mathcal{E}}
\newcommand{\Cend}{D_\mathcal{E}'}
\newcommand{\noUpdate}{\mathcal{F}}
\newcommand{\matchNB}{\mu}
\newcommand{\singleNB}{\sigma}
\newcommand\Exe{\mathcal{E}}
\newcommand\Voisin{N}

\newtheorem{definition}{Definition}
\newtheorem{lemma}{Lemma}
\newtheorem{theorem}{Theorem}
\newtheorem{corollary}{Corollary}

\newtheoremstyle{TheoremNum}
	{\topsep}{\topsep}         
	{\itshape}                      
	{}                                   
	{\bfseries}                     
	{.}                                  
	{ }                                  
	{\thmname{#1}\thmnote{ \bfseries #3}} 
\theoremstyle{TheoremNum}
\newtheorem{duplicateTheorem}{Theorem}
\newtheorem{duplicateCorollary}{Corollary}
\newtheorem{duplicateLemma}{Lemma}

\graphicspath{{./_figures/}}

\title{Polynomial self-stabilizing algorithm and proof for a 2/3-approximation of a maximum matching}
\date{}

\author[1]{Johanne Cohen}
\author[2]{Khaled Ma\^amra}
\author[1]{George Manoussakis}
\author[2]{Laurence Pilard}

\affil[1]{LRI-CNRS, Universit\'e Paris-Sud, Universit\'e Paris Saclay, France,

\texttt{\{johanne.cohen, george.manoussakis\}@lri.fr}}
\affil[2]{LI-PaRAD, Universit\'e Versailles-St. Quentin, Universit\'e Paris Saclay,  France, 
\texttt{\{khaled.maamra, laurence.pilard\}@uvsq.fr}}

\begin{document}
\maketitle

\begin{abstract}
We present the first polynomial self-stabilizing algorithm for finding a $\frac23$-approximation of a maximum matching in a general graph. The previous best known algorithm has been presented by Manne \emph{et al.} \cite{ManneMPT11} and has a sub-exponential time complexity under the distributed adversarial daemon \cite{Coor}. Our new algorithm is an adaptation of the Manne \emph{et al.} algorithm and works under the same daemon, but with a time complexity in $O(n^3)$ moves.  
Moreover, our algorithm only needs one more boolean variable  than the previous one, thus as in the Manne \emph{et al.} algorithm, it only requires a constant amount of memory space (three identifiers and $two$ booleans per node). 
\end{abstract}

\section{Introduction} 

In graph theory, a \emph{matching} $M$ in a graph $G$ is a subset of the edges of $G$ without common nodes.  A matching is \emph{maximal} if no proper superset of $M$ is also a matching whereas  a \emph{maximum} matching is a maximal matching with the highest cardinality among all possible maximal matchings.
Some (almost) linear time approximation algorithm for the maximum weighted matching problem   have been well studied \cite{DrakeH03,Preis1999}, nevertheless these algorithms are not distributed. They are based on  a simple greedy strategy using   \emph{augmenting path}.  An \emph{augmenting path} is a path, starting and ending in an unmatched node, and where every other edge is either unmatched or matched; \emph{i.e.} for each consecutive pair of edges, exactly one of them must belong to the matching. Let us consider the example in Figure~\ref{30eme}, page~\pageref{fig:correction}. In this figure, $u$ and $v$ are matched nodes and $x$, $y$ are unmatched nodes. The path $(x, u, v, y)$ is an augmenting path of length $3$ (written  \emph{$3$-augmenting path}). It is well known $\cite{doi:10.1137/0202019}$ that given a graph $G=(V,E)$ and a matching $M\subseteq E$,  if there is no augmenting path of length $2k-1$ or less, then $M$ is a $\frac{k}{k+1}-$approximation of the maximum matching. See \cite{DrakeH03} for the weighted version of this theorem. The greedy strategy in \cite{DrakeH03,Preis1999} consists in finding all augmenting paths of length $\ell$ or less and by switching  matched and unmatched edges of these paths in order to improve the maximum matching approximation.

In this paper, we present a self-stabilizing algorithm for finding a maximum matching with approximation ratio $2/3$ that uses the greedy strategy presented above. Our algorithm stabilizes after $O(n^{3})$ moves under the adversarial distributed daemon.

\section{Model} \label{sec:model}

The system consists of a set of processes where two adjacent processes can communicate with each other. The communication relation is represented by an undirected graph $G = (V, E)$ where $|V | = n$ and $|E| = m$. Each process corresponds to a node in $V$ and two processes $u$ and $v$ are adjacent if and only if $(u,v)\in E$. The set of neighbors of a process $u$ is denoted by $\Voisin(u)$ and is the set of all processes adjacent to $u$, and $\Delta$ is the maximum degree of $G$. We assume all nodes in the system have a unique identifier.

For the communication, we consider the \emph{shared memory model}.  In this model, each process maintains a set of \emph{local variables} that makes up the \emph{local state} of the process.   A process can read its local variables and the local variables of its neighbors, but it can write only in  its own local variables. A \emph{configuration} $C$ is the local states of all processes in the system.   Each process executes the same algorithm that consists of a set of \emph{rules}. Each rule is of the  form of ${<}name{>}::\text{ \textbf{if} } {<}guard{>} \text{ \textbf{then} } {<}command{>}$.  The \emph{name} is the name of the rule.  The  \emph{guard} is a predicate over the variables of both the process and its neighbors. The  \emph{command} is a sequence of actions assigning new values to the local variables of the process.
 
A rule is \emph{activable} in a configuration $C$ if its guard in $C$ is true. A process is \emph{eligible} for the rule $\mathcal{R}$ in a configuration $C$ if   its rule $\mathcal{R}$ is activable  in $C$ and we say the process is \emph{activable} in $C$.  An \emph{execution} is an alternate sequence of configurations and actions  $\Exe = C_0,A_0,\ldots,C_i,A_i, \ldots$, such that $\forall i\in \mathbb{N}^*$, $C_{i+1}$ is obtained by executing the command of at least one rule
that is activable  in $C_i$ (a process that executes such a rule makes a \emph{move}). More precisely, $A_i$ is the non empty set of activable  rules in $C_i$ that has been executed to reach $C_{i+1}$ and such that each process has at most one of its rules in $A_i$. We use the notation $C_{i} \mapsto C_{i+1}$ to denote this transition in $\Exe$.
Finally, let $\Exe'=C'_{0},A'_{0},\cdots,C'_{k}$ be a finite execution. We say $\Exe'$ is a \emph{sub-execution} of $\Exe$ if and only if   $\exists t\geq 0$ such that $\forall j \in [0,\cdots,k]$:($C'_{j}=C_{j+t} \land A'_{j}=A_{j+t}$).

An \emph{atomic operation} is such that no change can take place during its run, we usually assume that an atomic operation is instantaneous.
In the shared memory model, a process $u$ can read the local state of all its neighbors and update its whole local state in one atomic step. Then, we assume here that a rule is an atomic operation. An execution is \emph{maximal} if it is infinite, or it is finite and no process is activable in the last configuration. All algorithm executions considered here are assumed to be maximal.

A \emph{daemon} is a predicate on the executions. We consider only the most powerful one: the \emph{adversarial distributed daemon} that allows all executions described in the previous paragraph. Observe that we do not make any fairness assumption on the executions. 

An algorithm is \emph{self-stabilizing} for a given specification, if there exists a sub-set $\cal L$ of the set of all configurations such that: every execution starting from a configuration of $\cal L$ verifies the specification  (\emph{correctness})  and  starting from any configuration, every execution eventually reaches a configuration of $\cal L$   (\emph{convergence}). $\cal L$ is called the set of \emph{legitimate configurations}.  A configuration is \emph{stable} if no process is activable in the configuration. The algorithm presented here, is \emph{silent}, meaning that once the algorithm has stabilized, no process is activable. In other words, all executions of a silent algorithm are finite and end in a stable configuration. Note the difference with a non silent self-stabilizing algorithm that has at least one infinite execution with a suffix only containing legitimate configurations, but not stable ones.

\section{Algorithm} \label{sec:algo}

The algorithm presented in this paper is called $\algoName$, and  is based on the  algorithm presented by Manne \emph{et al.}~\cite{ManneMPT11}. As in the Manne \emph{et al.} algorithm, $\algoName$  assumes there exists an underlying maximal matching algorithm, which has reached a stable configuration. 
Based on  this stable maximal matching,  $\algoName$ builds a $\frac23$-approximation of the maximum matching by detecting and then deleting all $3$-augmenting paths. Once a $3$-augmenting path is detected, nodes rearrange the matching accordingly, \emph{i.e.}, transform this path with one matched edge into a path with two matched edges. This transformation leads to the deletion of the augmenting path and increases by one the cardinality of the matching. The algorithm   stabilizes when there is no augmenting path of length three left. By the result of Hopcroft \emph{et al.} \cite{doi:10.1137/0202019}, we obtain a $\frac23$-approximation of the maximum matching.

This underlying stabilized maximal matching can be built, for instance, with the self-stabilizing maximal matching algorithm from \cite{ManneMPT09} that stabilizes in $O(m)$ moves under the adversarial distributed daemon (so the same daemon than the one used in this paper). Observe that this algorithm is silent, meaning that the maximal matching remains constant once the algorithm has stabilized. Then, using a classical composition of this two algorithms \cite{Dolev00}, we obtain a total time complexity in $O(n^2\times n^3)= O(n^5)$ moves under the adversarial distributed daemon.

In the rest of the paper,  $\M$  is the underlying maximal matching, and $\MM$ is the set of edges built by our algorithm $\algoName$ (see Definition \ref{defMM}). 
For a set of nodes $A$, we define $single(A)$ and $matched(A)$ as the set of unmatched and matched nodes in $A$, accordingly to the underlying maximal matching $\M$. Moreover, $\M$ is encoded with the variable $m_u$. 
If $(u,v)\in \M$ then $u$ and $v$ are \emph{matched nodes} and we have: $m_u = v \land m_v = u$. If $u$ is not incident to any edge in $\M$, then $u$ is a \emph{single node} and $m_u=null$.
Since we assume the underlying maximal matching is stable, a node membership in $matched(V)$ or $single(V)$ will not change, and each node $u$ can use the value of $m_u$ to determine which set it belongs to.\\

\noindent\textbf{Variables description:}
In order to facilitate the rematching, each node $u \in V$ maintains three pointers and two boolean variables. The pointer $p_u$ refers to a neighbor of $u$ that $u$ is trying to (re)match with. If $p_u = null$ then the matching of $u$ has not changed from the maximal matching. Thus, the matching $\MM$ built by our algorithm is defined as follows:

\begin{definition} \label{defMM}
The set of edges built by  algorithm \algoName is 
$\MM = \{ (u,v) \in \M : p_u=p_v=null \} \cup \{ (a,b) \in E \setminus \M : p_a=b \land p_b=a\}$
\end{definition}

For a matched node $u$, pointers $\alpha_u$ and $\beta_u$ refer to two nodes in $single(N(u))$ that are \emph{candidates} for a possible rematching with $u$. 
Also, $s_u$ is a boolean variable that indicates if the node  $u$ has performed a successful rematching with its single node candidate. Finally, $end_u$  is a boolean variable that indicates if both $u$ and $m_u$ have performed a successful rematching or not.
For a single node $x$, only one pointer $p_x$ and one boolean variable $end_x$ are needed. $p_x$ has the same purpose as the $p$-variable of a matched node. 
The $end$-variable of a single node allows the matched nodes to know whether it is \emph{available} or not. A single node is \emph{available} if it is possible for this node to eventually rematch with one of its neighboring married node, \emph{i.e.}, $end_{x}=$\textit{False}. 

In our algorithm, $Unique(A)$ returns the number of unique elements in the multi-set $A$, and $Lowest(A)$ returns the node in $A$ with the lowest identifier. If $A=\emptyset$, then  $Lowest(A)$ returns $null$. Moreover, rules have priorities. In the algorithm, we present rules from the highest priority (at the top) to the lowest one (at the bottom).\\

\begin{figure}
\input{algo}
\end{figure}

\noindent\textbf{Graphical convention:}
We will follow the above conventions in all the figures: 
matched nodes are represented with double circles and single nodes with simple circles. Moreover, all edges that belong to the maximal matching $\M$ are represented with a double line, whereas the other edges are represented with a simple line. Black arrows show the content of the local variable $p$. If the $p$-value is $null$, we draw a \textsf{'T'}. A prohibited value is first drawn in grey, then scratched out in black.  If there is no knowledge on the $p$-value, nothing is drawn.  
For instance, in Figure \ref{3eme}, page~\pageref{3eme}, $x$ is a single node, $u$ and $v$ are matched nodes and $(u,v)\in \M$,  $p_u = x$, and $p_{x}\neq u$. In Figure \ref{30eme} page~\pageref{30eme}, $p_{u}=\bot$.

~\\
Now, we present the proof of our algorithm.

\section{Correctness Proof}

We first introduce some notations. A matched node $u$ is said to be \emph{First} if $AskFirt(u) \neq null$. In the same way, $u$ is \emph{Second} if $AskSecond(u) \neq null$.
Let $\Match : V \to V \cup \{null\}$ be a function where 
$\Match(u) =  AskFirst(u) $ if $AskFirst(u)\neq null$, otherwise $\Match(u) =AskSecond(u)  $. We will say a node makes a \emph{match} rule if it performs a $MatchFirst$ or $MatchSecond$ rule.

\textcolor{black}{If $C$ and $C'$ are two configurations in $\Exe$, then we note $  C \leq C'$ if and only if $C$ appears  before $C'$  in $\Exe$ or if $  C = C'$.  Moreover, we  write $\Exe \backslash C$ to denote all configurations of $\Exe$ except configuration $C$.}

\textcolor{black}{
\begin{definition}
Let $G=(V,E)$ be a graph and $M$ be a maximal matching of $G$. $(x, u, v, y)$ is a \emph{$3$-augmenting path} on $(G,M)$ if: 
\begin{enumerate}
\itemshort
\item $(x, u, v, y)$ is a path in $G$ (so all nodes are distincts);
\item $\{ (x,u), (v,y) \} \subset E\setminus  M$;
\item $(u,v) \in M$
\end{enumerate}
\end{definition}
\noindent For instance, in Figure \ref{30eme}, $(x, u, v, y)$ is a $3$-augmenting path.\\}

Recall that the set of edges built by our algorithm \algoName is 
$\MM = \{ (u,v) \in \M : p_u=p_v=null \} \cup \{ (a,b) \in E \setminus \M : p_a=b \land p_b=a\}$.

For the correctness part of the proof, we prove that in a stable configuration, $\MM$ is a $2/3$-approximation of a maximum matching on graph $G$. To do that we demonstrate there is no 3-augmenting path on $(G,\MM)$. In particular we prove  that for any  edge $(u,v) \in \M$, we have either $p_u = p_v = null$, or $u$ and $v$ have two distincts single neighbors they are rematched with, \emph{i.e.}, $\exists x\in single(N(u)), \exists y\in single(N(v))$ with $x\neq y$ such that $(p_{x}=u) \wedge (p_{u}=x) \wedge (p_{y}=v) \wedge (p_v=y)$. In order to prove that, we show  every other case for $(u,v)$ is impossible.  Main studied cases are shown in Figure \ref{fig:correction}. Finally, we prove that if $p_u = p_v = null$ then $(u,v)$ does not belong to a 3-augmenting-path on $(G,\MM)$.

\noindent \begin{figure}[h!]
\begin{center}
\hfill
\subfloat[By Lemma \ref{lem:3}]{
	\includegraphics[scale=0.34]{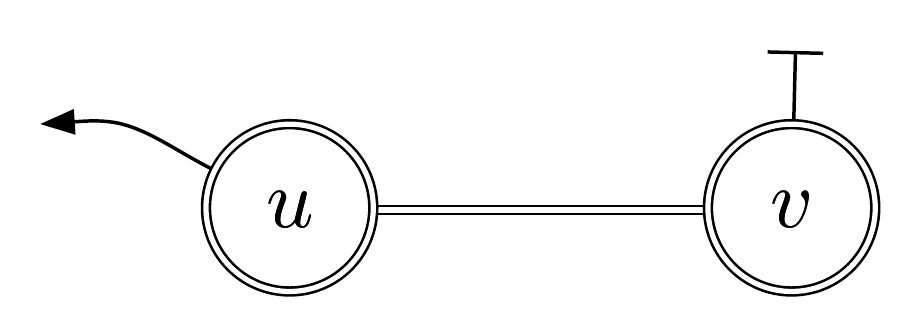}
	\label{2eme}
}\hfill
\subfloat[By Lemma \ref{2}]{
	\includegraphics[scale=0.34]{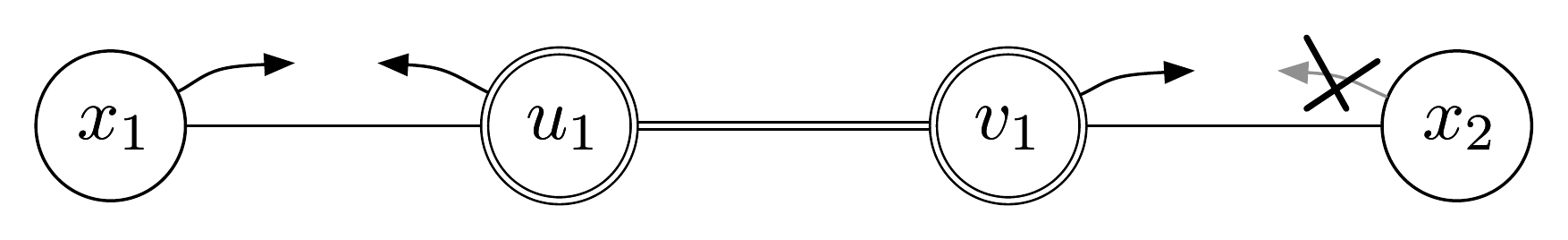}
	\label{fig:2}
}\hfill
\subfloat[By Lemma \ref{lem:utile}]{
	\includegraphics[scale=0.34]{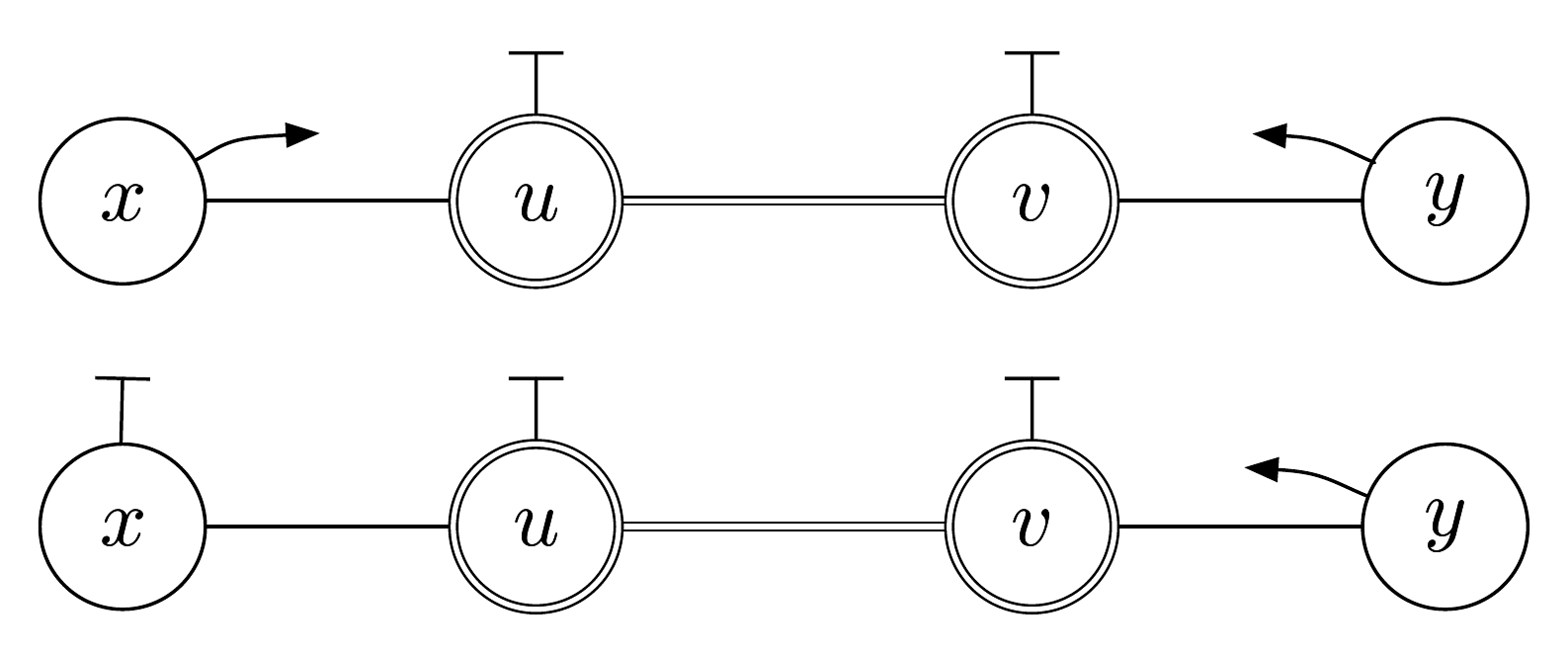}
	\label{1er}
}\hfill
\\
\hfill
\subfloat[By Lemma \ref{4}]{
	\includegraphics[scale=0.35]{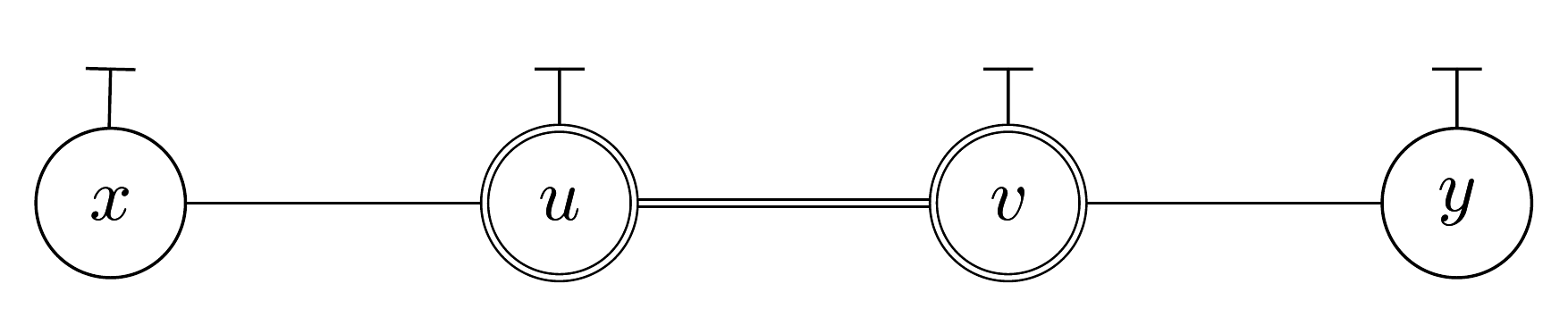}
	\label{30eme}
}\hfill
\subfloat[By Lemma \ref{5}]{
	\includegraphics[scale=0.35]{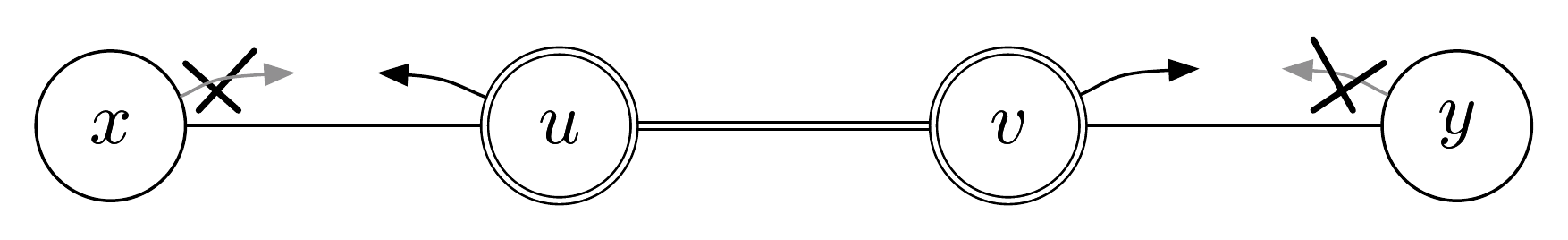}
	\label{3eme}
}\hfill
\caption{Impossible situations in a stable configuration.}
\label{fig:correction}
\end{center}
\end{figure}


\begin{lemma}\label{lem:trivial} 
In any stable configuration, we have the following properties: 
\begin{itemize}
\itemshort
\item  $\forall u \in matched(V) : p_u = \Match(u)$; 
\item  $\forall x \in single(V)$ : if $p_x=u$ with $u\neq null$, then $u \in matched(N(x)) \land p_u=x \land end_u=end_x$.
\end{itemize}
\end{lemma}

\begin{proof}
First, we will  prove the first property. We consider the case where  $AskFirst(u) \neq null$.  We have   $p_{u} = AskFirst(u) $, otherwise node $u$ can execute rule $AskFirst$. We can apply the same result for the case where $AskSecond(u) \neq null$. Finally, we consider the case where  $AskFirst(u) = AskSecond(u) = null$. If  $p_{u}\neq null$, then node $u$ can execute rule $ResetMatch$ which yields the contradiction. Thus, $p_{u}= null$.

Second, we consider a stable configuration $C$ where $p_{x}=u$, with $u\neq null$. $u\in matched(N(x))$, otherwise $x$ is eligible for an \emph{UpdateP} rule. Now there are two cases: $p_{u}=x $ and $p_{u}\neq x$.  If $p_{u}\neq x$, this means that $p_{p_{x}}\neq x$. Thus, $x$ is eligible for rule $UpdateP$, and this yields to a contradiction with the fact  that $C$ is stable.  
Finally,  we have $end_{u}=end_{x}$, otherwise  $x$ is eligible for rule $UpdateEnd$.
\end{proof}

 
\begin{lemma}\label{lem:3}
 Let $(u, v)$ be an edge in $\M$. Let $C$ be a configuration. If $p_u \neq null  \land p_v=null$ holds in $C$ (see Figure \ref{2eme}), then $C$ is not stable.
\end{lemma} 

\begin{proof}
By contraction. We assume  $C$ is  stable.  From Lemma~\ref{lem:trivial},  we have $p_u=\Match(u) \neq null$ and $p_v=\Match(v)$. So, by definition of predicates $AskFirst$ and  $AskSecond$,  $\Match(u)=x\neq null$ implies that $\Match(v) \neq  null$. This contradicts that fact that  $p_v=\Match(v) =null$.
\end{proof}


\begin{lemma}\label{lem:augmentingpath} 
Let $(x, u, v, y)$ be a 3-augmenting path on $(G,\M)$. Let $C$ be a stable configuration. 
In C, if $p_x=u$, $p_u=x$, $p_v=y$ and $p_y = u$, then $end_x= end_u=end_v=end_y=True$.
\end{lemma}

\begin{proof}
From Lemma~\ref{lem:trivial}, $p_u=\Match(u)$ (resp. $p_v=\Match(v)$)  thus $\Match(u)\neq null$ and  $\Match(v)\neq null$. 
W.l.o.g, we can assume that $AskFirst(u) \neq null$.  We have $s_u=True$, otherwise $u$ can execute $MatchFirst$ rule. Now, as $s_u=True$, we must have $end_v=True$, otherwise $v$ can execute $MatchSecond$ rule. As $s_u=end_v=True$, we must have $end_u=True$, otherwise $u$ can execute $MatchFirst$ rule. From Lemma~\ref{lem:trivial}, we can deduce that $end_x= end_u=end_v=end_y=True$ and this concludes the proof.
\end{proof}


\begin{lemma}\label{2}
Let $(x_1, u_1, v_1, x_2)$ be a 3-augmenting path on $(G,\M)$. Let $C$ be a  configuration. 
If $p_{x_1}=u_1 \land p_{u_1}= x_1  \land p_{v_1}=x_2 \land p_{x_2}\neq v_1$ holds in $C$ (see Figure \ref{fig:2}), then $C$ is not stable.
\end{lemma}

\begin{proof}
By contraction. We assume $C$ is  stable. From Lemma~\ref{lem:trivial}, $\Match(u_1)=x_{1}$ and  $\Match(v_1)=x_{2}$.

First we assume that $AskSecond(u_1)=x_{1}$ and $AskFirst(v_1)=x_{2}$. The local variable $s_{v_1}$ is $False$,  otherwise $v_1$ would be eligible for executing the $MatchFirst$ rule. Since $AskSecond(u_1)\neq null \wedge p_{u_1} \neq null \wedge s_{v_1} =False$, this implies that $u_1$ is eligible for the $ResetMatch$ rule which is a contradiction.

Second, we assume that $AskFirst(u_1)=x_{1}$ and $AskSecond(v_1)=x_{2}$.  We have $s_{u_1} = True$, otherwise $u_1$ can execute the $MatchFirst$ rule. This implies that $end_{v_1} = False$, otherwise $v_1$ can execute the $MatchSecond$ rule. As $end_{v_1} = False$, then $end_{u_1} = False$, otherwise $u_1$ can execute the $MatchFirst$ rule. From Lemma~\ref{lem:trivial},   $end_{x_{1}}= end_{u_1}=end_{v_1}=False$. Since $\Match(v_1)=x_{2}$, we have $x_{2}\in \{\alpha_{v_1},\beta_{v_1}\}$. Let us assume $end_{x_2}=True$. Then  $x_2\not\in BestRematch(v_1)$ and then $v_1$ is elligible for an $Update$. Thus $end_{x_2}=False$. 

Therefore, $C$ is a configuration such that $u_1$ is $First$ and $v_1$ is $Second$ with $end_{x_{1}}=end_{u_1}=end_{v_1}=end_{x_{2}}=False$. Now we are going to show  there exists another augmenting path $(x_2, u_2, v_2, x_3)$ with $end_{x_2}=end_{u_2}=end_{v_2}=end_{x_3}=False$ and $p_{u_2}=x_{2}$, $p_{x_2}=u_{2}$, $p_{v_{2}}=x_{3}$ and $p_{x_3}\neq v_2$ such that $u_{2}$ is $First$ and $v_{2}$ is $Second$ (see Figure~\ref{etape}).

\begin{figure}[h]
\begin{center}
\includegraphics[scale=0.35]{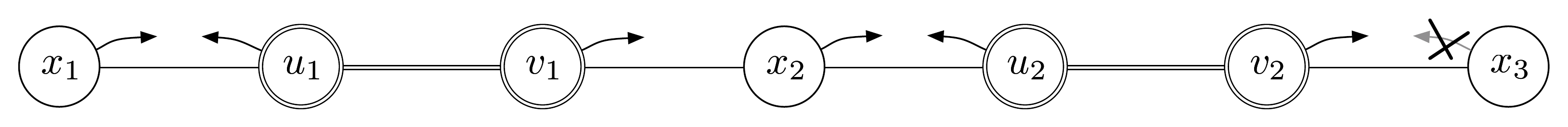}
\caption{A chain of 3-augmenting paths.}
\label{etape}
\end{center}
\end{figure}

$p_{x_2}\neq null$ otherwise $x_2$ is elligible for an \emph{UpdateP} rule. Thus there exists a vertex $u_2 \neq v_1$ such that  $p_{x_{2}}=u_2$. From Lemma~\ref{lem:trivial}, $u_2\in matched(N(x_2))$ and $p_{u_2}=x_{2}$. Therefore, there exists a node $v_2=m_{u_2}$. From Lemma~\ref{lem:3}, we can deduce that $p_{v_2} \neq null$ and there exists a node $x_3$ such that $p_{v_2} = x_{3}$. $x_3\in single(N(v_2))$ otherwise $x_2$ is eligible for an \emph{Update} rule. Finally, if $p_{x_{3}} = v_2 $, then  Lemma~\ref{lem:augmentingpath} implies that $end_{x_2}=end_{a_2}=end_{b_2}=end_{x_3}=True$.  This yields to the contradiction with the fact $end_{x_{2}}=False$. So, we have $p_{x_{3}} \neq v_2 $. 

We can then conclude that $(x_2, u_2, v_2, x_3)$ is a 3-augmenting path such that $p_{x_2}=u_2 \land p_{u_2}= x_2  \land p_{v_2}=x_3 \land p_{x_3}\neq v_2$. This augmenting path has the exact same properties than the first considered augmenting path $(x_1, u_1, v_1, x_2)$ and in particular $u_1$ is First. 

Now we can continue the construction in the same way. Therefore, for $C$ to be stable, it has to exist a chain of 3-augmenting paths $(x_1, u_1, v_1, x_2, u_2, v_2, x_3, \ldots, x_i, u_i, v_i, x_{i+1}, \ldots)$ where $\forall i \geq 1: (x_i, u_i, v_i, x_{i+1})$ is a 3-augmenting path with  $p_{x_i}=u_i \land p_{u_i}= x_i  \land p_{v_i}=x_{i+1} \land p_{x_{i+1}}= v_{i+1}$  and $u_i$ is First. Thus, $x_{1} < x_{2} < \ldots < x_{i} <\dots$ since the $u_i$ will always be First. Since the graph is finite some $x_{k}$ must be equal to some $x_{\ell}$ with $\ell \neq k$ which contradicts the fact that the identifier' sequence is strictly increasing.
\end{proof}


\begin{lemma}
\label{5}
Let $(x, u, v, y)$ be a 3-augmenting path on $(G,\M)$. Let $C$ be a configuration. If $p_u=x \land p_x \neq u  \land p_v=y \land p_y \neq v$ holds in $C$ (see Figure \ref{3eme}), then $C$ is not stable.
\end{lemma}

\begin{proof}
By contradiction, assume $C$ is stable.  From Lemma~\ref{lem:trivial}, $\Match(u)= x$. Assume to begin that $AskFirst(u) \neq null$. Because $p_{p_{u}} \neq u$ we have $s_{u} = False$, otherwise  $u$ is eligible for $MatchFirst$. Since $AskSecond(v)\neq null$ and $s_{m_{v}}  =s_{u} = False$ then $v$ can apply the $ResetMatch$ rule which yields a contradiction. Therefore assume that $AskSecond(u) \neq null$. The situation is symmetric (because now $AskFirst(v)\neq null$) and therefore we get the same contradiction as before.
\end{proof}


\begin{lemma}\label{4}
Let $(x, u, v, y)$ be a 3-augmenting path on $(G,\M)$. Let $C$ be a configuration. If $p_{y}= p_u=p_v=p_{y}=null$  holds in $C$ (see Figure \ref{30eme}), then $C$ is not stable.
\end{lemma}

\begin{proof}
By contradiction, assume $C$ is stable. $end_x=False$ (resp. $end_y=False$), otherwise $x$ (resp. $y$) is eligible for a \emph{ResetMatch}.  $(\alpha_u,\beta_u)=BestRematch(u)$ (resp. $(\alpha_v,\beta_v)=BestRematch(v)$), otherwise $u$ (resp. $v$) is eligible for an \emph{Update}. Thus, there is at least an available single node for $u$ and $v$ and so $\Match(u)\neq null$ and $\Match(v)\neq null$. Then, this contradicts the fact that  $\Match(u)= null$ (see Lemma~\ref{lem:trivial}).
\end{proof}


\begin{theorem}
\label{corr}
In a stable configuration we have, $\forall(u,v)\in \M$:
\begin{itemize}
\itemshort
\item  $p_{u}=p_{v}=null$ or
\item $\exists x\in single(N(u)), \exists y\in single(N(v))$ with $x\neq y$ such that $p_{x}=u \wedge p_{u}=x \wedge p_{y}=v \wedge p_v=y$. 
\end{itemize}
\end{theorem}

\begin{proof}
We will prove that all cases but these two are not possible in a stable configuration.
First,  Lemma \ref{lem:3} says the configuration cannot be stable if exactly one of $p_u$ or $p_v$ is not $null$. 

Second, assume that  $p_u\neq null \land p_v \neq null$. Let $p_{u}=x$ and $p_{v}=y$. Observe that $x\in single(N(u))$ (resp. $y\in single(N(v))$), otherwise $u$ (resp. $v$) is eligible for \emph{Update}.

Case $x\neq y$: If $p_{x}\neq u $ and $p_{y}\neq v$ then Lemma \ref{5} says the configuration cannot be stable. If $p_{x} = u $ and $p_{y}\neq v$ then Lemma \ref{2} says the configuration cannot be stable. Thus, the only remaining possibility when $p_u\neq null$ and $p_v \neq null$ is: $p_{x} = u $ and $p_{y}= v$.

Case $x=y$: $Ask(u)\null$ (resp. $Ask(v)\neq null$), otherwise $u$ (resp. $v$) is eligible for a \emph{ResetMatch}. W.l.o.g. let us assume that $u$ is First. $x=AskFirst(u)$ (resp. x=AskSecond(v)), otherwise $u$ (resp. $v$) is eligible for \emph{MatchFirst} (resp. \emph{MatchSecond}). Thus $AskFirst(u) = AskSecond(v)$ which is impossible according to these two predicates.
\end{proof}


\begin{lemma} \label{lem:utile}
Let $x$ be a single node. In a stable configuration, if $p_x=u, u\neq null$ then it exists a 3-augmenting path $(x,u,v,y)$ on $(G,\M)$ such that $p_x=u \land p_u=x \land p_v=y \land p_y=v$.
\end{lemma}

\begin{proof}
By lemma \ref{lem:trivial}, if $p_x=u$ with $u\neq null$ then $u\in matched(N(x))$ and $p_u=x$. Since $p_u\neq null$, by Theorem \ref{corr} the result holds. 
\end{proof}

Observe that according to this Lemma, cases from Figure~\ref{1er} are impossible. \\


Thus, in a stable configuration, for all edges $(u,v)\in \M$, if $p_u = p_v = null$ then $(u,v)$ does not belong to a 3-augmenting-path on $(G,\MM)$. In other words, we obtain:

\begin{corollary}
In a stable configuration, there is no 3-augmenting path on $(G,\MM)$ left.
\end{corollary}

\section{Convergence Proof}

We start by giving a sketch of the convergence proof. In the following, $\matchNB$ denotes the number of matched nodes and $\singleNB$ the number of single nodes.

First, we focus on the variable $end$.  Recall that for a matched edge $(u,v)$, the variable $end$ indicates if both $u$ and $v$ have performed a successful rematch or not. In section \ref{section:part1}, we prove the maximum number of times a matched node $u$ can write $True$ in its  variable $end$ is $2$. The idea is that  only one writing can correspond to an incorrect initialization of the node $m_{u}$.  

\begin{duplicateTheorem}[\ref{nber}]
In any execution, a matched node $u$ can write  $end_{u} := True$ at most twice.
\end{duplicateTheorem}

This theorem is the key point of the convergence proof.  Indeed, this result exhibits an action that can be performed only a finite number of time. Thus, with this result, we have a first strong step leading to the proof of the silent property of our algorithm.\\

After considering the $end$ variable of matched nodes in section \ref{section:part1}, we focus on the $end$ variable of single nodes in section \ref{section:part2}. Since single nodes just follow orders from their neighboring matched nodes (Lemma~\ref{theol}),  
we can count the number of times single nodes can change the value of their $end$ variable as described above. 
There are $\singleNB$ possible modifications due to bad initializations. Moreover, a matched node $u$ can write  $True$ twice in  $end_{u}$, so $end_{u}$ can be $True$ during $3$ distinct sub-executions. As a single node $x$ copies the $end$-value of the matched node it points to ($p_{x}=u$), then a single node can write $True$ in its $end$-variable at most $3$ times as well. So we obtain $6\matchNB$ modifications.

\begin{duplicateLemma}[\ref{big}]
In any execution, the number of transitions where a single node changes the value of its $end$ variables (from $True$ to $False$ or from $False$ to $True$) is  at most $\singleNB+6 \matchNB$ times.
\end{duplicateLemma}

In section \ref{section:part3}, we count the maximal number of $Update$ rule that can be performed in any execution. To do that, we observe that the first line of the $Update$ guard can be $True$ at most once in an execution (Lemma~\ref{lem:UpdateL1}). Then we prove for the second line of the guard to be $True$, a single node has to change its $end$ value (Lemma \ref{flip}).  Thus, for each single node modification of the $end-$value, at most all matched neighbors  of this single node can perform one $Update$ rule.

\begin{duplicateCorollary}[\ref{update}]
Matched nodes can execute at most $\Delta(\singleNB+6 \matchNB) + \matchNB$ times the $Update$ rule.
\end{duplicateCorollary}

In section \ref{section:part4}, we count the maximal number of moves performed by matched nodes between two $Update$. The idea is that in an execution without  $Update$, $\alpha$ and $\beta$ values of all matched nodes remain  constant. Thus, in these small executions, at most one augmenting path is detected per matched edge and at most one rematch attempt is performed  per matched edge. We obtain that the maximal number of moves of a matched node  in these  small executions is  $12$ (Lemma~\ref{bupd}).  By the previous remark and Corollary~\ref{update},  we obtain: 

\begin{duplicateTheorem}[\ref{matched}]
In any execution, matched nodes can execute at most $12\matchNB  (\Delta(\singleNB+6 \matchNB) + \matchNB)$ rules.
\end{duplicateTheorem}

Finally, we count the maximal number of moves that single nodes can perform, counting rule by rule. The $ResetEnd$ is done at most once (Lemma \ref{lem:ResetEndMove}). The number of $UpdateEnd$ is bounded by the number of times single nodes can change their $end$-value, so it  is at most  $\singleNB + 6 \matchNB$ (Lemma~\ref{lem:updateEndMove}). Finally, $UpdateP$ is counted as follows: between two consecutive $UpdateP$ executed by a single node $x$, a matched node has to make a move.  The total number of executed \emph{UpdateP} is then at most  $12\matchNB (\Delta (\singleNB + 6 \matchNB) + \matchNB)+1$ (See Lemma~\ref{lem:updateMove}).  

\begin{duplicateCorollary}[\ref{cor:end}]
The algorithm $\algoName$ converges in $O(n^3)$ steps under the adversarial distributed daemon. 
\end{duplicateCorollary}

\subsection{A matched node can write $True$ in its $end$-variable at most twice}\label{section:part1}

The first three lemmas are technical lemmas.

\begin{lemma}\label{Trr}
Let $u$ be a matched node. Consider an execution $\Exe$ starting after $u$ executed some rule. Let $C$ be any configuration in $\Exe$. If $end_u=True$ in $C$ then $s_u=True$ as well.
\end{lemma}

\begin{proof}
Let $C_0\mapsto C_1$ be the transition in $\Exe$ in which $u$ executed a rule for the last time before $C$. Observe that $C$ may be equal to $C_1$. The executed rule is necessarily a $match$ rule, otherwise $end_u$ could not be $True$ in $C_1$. If it is a $MatchSecond$ the lemma holds since in that case $s_u$ is a copy of $end_u$. Assume now it is a $MatchFirst$. For $end_u$ to be $True$ in $C_1$,  $p_u=AskFirst(u) \land p_{p_u} = u \land p_{m_u}=AskSecond(m_u)$ must hold in $C_0$, according to the guard of \emph{MatchFirst}.  This implies that $u$ writes $True$ in $s_u$ in transition $C_0\mapsto C_1$.
\end{proof}


\begin{lemma}\label{ess}
Let $u$ be a matched node. Consider an execution $\Exe$ starting after $u$ executed some rule. Let $C$ be any configuration in $\Exe$. In $C$, if $s_u=True$ then $\exists x\in single(N(u)) : p_u =x \land p_x = u$.
\end{lemma}

\begin{proof}
Consider transition $C_0\mapsto C_1$ in which $u$ executed a rule for the last time before $C$. The executed rule is necessarily a $match$ rule, otherwise $s_u$ could not be $True$ in $C_1$. Observe now that whichever $match$ rule is applied,  $Ask(u)\neq null$ -- let us assume $Ask(u)=x$ -- and $p_u=x$  and $p_x=u$ must hold in $C_0$ for $s_u$ to be $True$ in $C_1$. $p_u=x$ still holds in $C_1$ and until $C$. Moreover, $x$ must be in $single(N(u))$, otherwise $u$ would have executed an \emph{Update} instead of a match rule in $C_0\mapsto C_1$, since \emph{Update} has the higest priority among all rules. Finally, in transition $C_0\mapsto C_1$, $x$ cannot execute $UpdateP$ nor $ResetEnd$ since $p_x\in matched(N(x)) \land p_{p_x}=x$ holds in $C_0$. Thus in $C_1$, $p_u=x$ and $p_x=u$ holds. Using the same argument, $x$ cannot execute $UpdateP$ nor $ResetEnd$ between configurations $C_1$ and $C$. Thus $p_u=x \land p_x=u$ in $C$.
\end{proof}


\begin{lemma}\label{lem:UpdateL1}
Let $u$ be a matched node  and $\Exe$ be an execution containing a transition $C_0\mapsto C_1$ where $u$ makes a move. From $C_1$, the predicate in the first line of the guard of the $Update$ rule will ever hold from $C_1$. 
\end{lemma}

\begin{proof}
Let $C_2$ be any configuration in $\Exe$ such that $C_2 \geq C_1$. Let $C_{10} \mapsto C_{11}$ be the last transition before $C_2$ in which $u$ executes a move. Notice that by definition of $\Exe$, this transition exists. Assume by contradiction that one of the following predicates holds in $C_2$.

\begin{enumerate}
\itemshort
\item $ (\alpha_u > \beta_u) \vee  (\alpha_u,\beta_u \notin (single(N(u)) \cup \{null \})) \vee (\alpha_u =\beta_u \land \alpha_u \neq null)$
\item $p_u \notin (single(N(u))\cup \{null \})$
\end{enumerate}

By definition between $C_{11}$ and $C_2$, $u$ does not execute rules. To modify the variables $\alpha_u, \beta_u$ and $p_u$, $u$ must execute a rule. Thus one of the two predicates also holds in $C_{11}$.

We first show that if predicate (1) holds in $C_{11}$ then we get a contradiction. If $u$ executes an $Update$ rule in transition $C_{10} \mapsto C_{11}$, then by definition of the $BestRematch$ function, predicate (1) cannot hold in $C_{11}$ (observe that the only way for $\alpha_u=\beta_u$ is when $\alpha_u = \beta_u = null$). Thus assume that $u$ executes a $match$ or $ResetMatch$ rule. Notice that these rules do not modify the value of the $\alpha_u$ and $\beta_u$ variables. This implies that if $u$ executes one of these rules in $C_{10} \mapsto C_{11}$, predicate (1) not only hold in $C_{11}$ but also in $C_{10}$. Observe that this implies, in that case that $u$ is eligible for $Update$ in $C_{10} \mapsto C_{11}$, which gives the contradiction since $Update$ is the rule with the highest priority among all rules.

Now assume predicate (2) holds in $C_{11}$. In transition $C_{10} \mapsto C_{11}$, $u$ cannot execute $Update$ nor $ResetMatch$ as this would imply that $p_u=null$ in $C_{11}$. Assume that in $C_{10} \mapsto C_{11}$ $u$ executes a $match$ rule. Since in $C_{11}$, $p_u \notin (single(N(u))\cup \{null \})$ this implies that in $C_{10}$,  $Ask(u) \notin (single(N(u))\cup \{null \})$. This implies that $\alpha_u, \beta_u \notin (single(N(u)) \cup \{null \})$ in $C_{10}$. Thus $u$ is eligible for $Update$ in transition $C_{10} \mapsto C_{11}$ and this yields the contradiction since $Update$ is the rule with the highest priority among all rules.

Since these two predicates cannot hold in $C_2$, this concludes the proof.
\end{proof}


Now, we focus on particular configurations for   a matched edge $(u,v)$ corresponding to the fact they have completely exploited a $3$-augmenting path.

\begin{lemma}\label{lemm:stableEnd}
Let $(u,v)$ be a matched edge, $\Exe$ be an execution and $C$ be a configuration of $\Exe$. If in $C$, we have: 
\begin{enumerate}
\itemshort
\item $p_u\in single(N(u)) \land p_u=AskFirst(u) \land p_{p_{u}}=u$;
\item $p_v\in single(N(v)) \land p_v=AskSecond(v) \land p_{p_{v}}=v$;
\item $s_u=end_u=s_v=end_v=True$;
\end{enumerate}
then neither $u$ nor $v$ will ever be eligible for any rule from $C$.
\end{lemma}

\begin{proof}
Observe first that neither $u$ nor $v$ are eligible for any rule in $C$. Moreover, $p_u$ (resp. $p_v$) is not eligible for an $UpdateP$ move since $u$ (resp. $v$) does not make any move. Thus $p_{p_u}$ and $p_{p_v}$ will remain constant since $u$ and $v$ do not make any move and so neither $u$ nor $v$ will ever be eligible for any rule from $C$.
\end{proof}


The  configuration $C$ described in Lemma~\ref{lemm:stableEnd} is called a $stop_{uv}$ configuration. From such a configuration neither $u$ nor $v$ will ever be eligible for any rule.

In Lemmas~\ref{lemm:stableFirst} and~\ref{lemm:stableSecond}, we consider executions where a matched node $u$ writes $True$ in $end_{u}$ twice, and we focus on the transition $C_0\mapsto C_1$ where $u$ performs its second writing. Lemma~\ref{lemm:stableFirst} shows that, if $u$ is First in $C_0$, then $C_{1}$ is a $stop_{um_{u}}$ configuration.  Lemma~\ref{lemm:stableSecond} shows that, if $u$ is Second in $C_0$, then either $C_{1}$ is a $stop_{um_{u}}$ configuration or it exists a configuration $C_{3}$ such that $C_{3}>C_{1}$, $u$ does not make any move from $C_{1}$ to $C_{3}$ and $C_{3}$ is a $stop_{um_{u}}$ configuration. 

Lemma~\ref{lem:uv:end} and  Corollary~\ref{upd} are  required to prove Lemmas~\ref{lemm:stableFirst} and~\ref{lemm:stableSecond}.


\begin{lemma}\label{lem:uv:end}
 Let $(u,v)$ be a matched edge. Let $\Exe$ be some execution in which $v$ does not execute any rule. If it exists a transition $C_0\mapsto C_1$ in $\Exe$ where $u$ writes $True$ in $end_u$, then $u$ is not eligible for any rule from $C_1$.
\end{lemma}

\begin{proof}
To write $True$ in $end_u$ in transition $C_0\mapsto C_1$, $u$ must have executed a $match$ rule. 
According to this rule, \mbox{$(p_u=Ask(u)\land p_{p_{u}}=u)$} holds $C_0$ with $p_u\in single(N(u))$, otherwise $u$ would have executed an \emph{Update} instead of a $match$ rule. Now, in $C_0\mapsto C_1$, $p_u$ cannot execute \emph{UpdateP} then it cannot change its $p$-value and $v$ does not execute any move then it cannot change $Ask(u)$. Thus, \mbox{$(p_u=Ask(u)\land p_{p_{u}}=u)$} holds in both $C_0$ and $C_1$.

Assume now by contradiction that $u$ executes a rule after configuration $C_1$.  Let $C_2\mapsto C_3$ be the next transition in which it executes a rule. Recall that between configurations $C_1$ and $C_2$ both $u$ and $v$ do not execute rules. Observe also that $p_u$ is not eligible for $UpdateP$ between these configurations. Thus \mbox{$(p_u=Ask(u)\land p_{p_{u}}=u)$} holds from $C_0$ to $C_2$. Moreover the following points hold as well between $C_0$ and $C_2$ since in $C_0\mapsto C_1$ $u$ executed a $match$ rule and $v$ does not apply rules in $\Exe$:
\begin{itemize}
\itemshort
\item $\alpha_u$, $\alpha_v$, $\beta_u$ and $\beta_v$ do not change.
\item The values of the variables of $v$ do not change.
\item $Ask(u)$ and $Ask(v)$ do not change.
\item If $u$ was $First$ in $C_0$ it is $First$ in $C_2$ and the same holds if it was $Second$.
\end{itemize}

Using these remarks, we start by proving that $u$ is not eligible for $ResetMatch$ in $C_2$. If it is $First$ in $C_2$, this holds since $AskFirst(u)\neq null$ and $AskSecond(u)=null$. If it is $Second$ then to be eligible for $ResetMatch$, $s_v=False$ must hold in $C_2$ since $AskSecond(u)\neq null$. Since $u$ executed $end_u=True$ in $C_0\mapsto C_1$ and since $u$ was $Second$ in $C_0$, then necessarily $s_v=True$ in $C_0$ and thus in $C_2$ (using remark 2 above). So $u$ is not eligible for $ResetMatch$ in $C_2$.

We show now that $u$ is not eligible for an $Update$ in $C_2$. The $\alpha$ and $\beta$ variables of $u$ and $v$ remain constant between $C_0$ and $C_2$. Thus if any of the three first disjunctions in the $Update$ rule holds in $C_2$ then it also holds in $C_0$ and in $C_0\mapsto C_1$ $u$ should have executed an $Update$ since it has higher priority than the $match$ rules. Moreover since in $C_2$  \mbox{$(p_u=Ask(u)\land p_{p_{u}}=u)$} holds, the last two disjunctions of $Update$ are $False$ and we can state $u$ is not eligible for this rule.

We conclude the proof by showing that $u$ is not eligible for a $match$ rule in $C_2$. If $u$ was $First$ in $C_0$ then it is $First$ in $C_2$. To write $True$ in $end_u$ then $(p_u=AskFirst(u) \land p_{p_{u}} = u \land s_u \land p_{m_{u}} = AskSecond(m_u) \land end_{m_u})$ must hold in $C_0$. Since in $C_0\mapsto C_1$ $v$ does not execute rules, it also holds in $C_1$. 
The same remark between configurations $C_1$ and $C_2$ implies that this predicate holds in $C_2$. Thus in $C_2$, all the three conditions of the $MatchFirst$ guard are $False$ and $u$ not eligible for $MatchFirst$.
A similar remark  if $u$ is $Second$ implies that $u$ will not be eligible for $MatchSecond$ in $C_2$ if it was $Second$ in $C_0$.
\end{proof}


\begin{corollary}\label{upd}
Let $(u,v)$ be a matched edge. In any execution, if $u$ writes $True$ in $end_u$ twice, then $v$ executes a rule between these two writing. 
\end{corollary}


\begin{lemma}\label{lemm:stableFirst}
Let $(u,v)$ be a matched edge and $\Exe$ be an execution where $u$ writes $True$ in its variable $end_u$ at least twice. Let $C_0\mapsto C_1$ be the transition where $u$ writes $True$ in $end_u$ for the second time in $\Exe$. If $u$ is First in $C_0$ then the following holds:
\begin{enumerate}
\itemshort
\item in configuration $C_0$,
	\begin{enumerate}
	\itemshort
	\item $s_v=end_v=True$;
	\item $p_u = AskFirst(u) \land p_{p_u}=u \land s_u=True \land p_v = AskSecond(v)$;
	\item $p_u\in single(N(u))$;
	\item $p_v\in single(N(v)) \land p_{p_v}=v$; 
	\end{enumerate}
\item $v$ does not execute any move in $C_0\mapsto C_1$;
\item in configuration $C_1$, 
\begin{enumerate}
\itemshort
\item $s_u=end_u=True$;
\item $p_u\in single(N(u)) \land p_v\in single(N(v))$;
\item $s_v=end_v=True$;
\item $p_u = AskFirst(u) \land p_v = AskSecond(v)$; 
\item $p_{p_u}=u \land p_{p_v}=v$.
\end{enumerate}
\end{enumerate}
\end{lemma}

\begin{proof}
We prove Point 1a. Observe that for $u$  to write $True$ in $end_u$, $end_v$ must be $True$ in $C_0$. By Lemma \ref{Trr} this implies that $s_v$ is True as well. Now Point 1b holds by definition of the $MatchFirst$ rule.  As in $C_0$, $u$ already executed an action, then according to Lemma \ref{lem:UpdateL1}, Point 1c holds and will always hold. 
By Corollary \ref{upd}, $u$ cannot write $True$ consecutively if $v$ does not execute moves. Thus at some point before $C_0$, $v$ applied some rule. This implies that in configuration $C_0$, since s$_v=True$, by Lemma \ref{ess}, $\exists x\in single(N(v)): p_v =x \land p_x=v$. Thus Point 1d holds.

We now show that $v$ does not execute any move in $C_0\mapsto C_1$ (Point 2).
Recall that $v$ already executed an action before $C_0$, so by Lemma \ref{lem:UpdateL1}, line 1 of the $Update$ guard does not hold in $C_0$. Moreover, by Point 1d, line 2 does not hold either. Thus, $v$ is not eligible for $Update$ in $C_0$. We also have that $s_u=True$ and $AskSecond(v)\neq null$ in $C_0$, thus $v$ is not eligible for $ResetMatch$. Observe now that by Points 1a, 1b and 1d, $v$ is not eligible for $MatchSecond$ in $C_0$. Finally $v$ cannot execute $MatchFirst$ since $AskFirst(v)=null$. Thus $v$ does not execute any move in $C_0\mapsto C_1$ and so Point 2 holds. 

In $C_1$, $end_u$ is True by hypothesis and according to Point 1b, $u$ writes $True$ in $s_u$ in transition $C_0\mapsto C_1$. Thus Point 3a holds. Points 3b holds by Points 1c and 1d. Points 3c holds by Points 1a and 2. 
$AskFirst(u)$ and $AskSecond(v)$ remain constant in $C_0\mapsto C_1$ since neither $u$ nor $v$ executes an $Update$ in this transition. Moreover $p_v$ remains constant in $C_0\mapsto C_1$  by Point 2 and $p_u$ remains constant also since it writes $AskFirst(u)$ in $p_u$ in this transition while $p_u=AskFirst(u)$ in $C_0$. Thus Points 3d holds. Observe that nor $p_u$ neither $p_v$ is eligible for an $UpdateP$ in $C_0$, thus Point 3e holds.
\end{proof}


Now, we consider the case where $u$ is Second.

\begin{lemma}\label{lemm:stableSecond}
Let $(u,v)$ be a matched edge and $\Exe$ be an execution where $u$ writes $True$ in its variable $end_u$ at least twice. Let $C_0\mapsto C_1$ be the transition where $u$ writes $True$ in $end_u$ for the second time in $\Exe$. If $u$ is Second in $C_0$ then the following holds:
\begin{enumerate}
\itemshort
\item in configuration $C_0$,
	\begin{enumerate}
	\itemshort
	\item $s_v=True \land p_v=AskFirst(v)$;
	\item $p_v\in single(N(v)) \land p_{p_v} =v$;
	\end{enumerate}
\item in transition $C_0 \mapsto C_1$, $v$ is not eligible for $Update$ nor $ResetMatch$;
\item in configuration $C_1$,
	\begin{enumerate}
	\itemshort
	\item $s_u=end_u=True$;
	\item $p_v\in single(N(v)) \land p_v=AskFirst(v) \land p_{p_{v}}=v$;
	\item $p_u\in single(N(u)) \land p_u=AskSecond(u) \land p_{p_{u}}=u$;
	\item $s_v = True$;
	\end{enumerate}
\item $u$ is not eligible for any move in $C_1$;
\item If $end_u=False$  in $C_1$ then the following holds:
\begin{enumerate}
\itemshort
\item From $C_1$, $v$ executes a next move and this move is a $MatchFirst$;
\item Let us assume this move (the first move of $v$ from $C_1$) is done in transition $C_2 \mapsto C_3$.  In configuration $C_3$, we have:
	\begin{enumerate}
	\itemshort
	\item $s_u=end_u=True$;
	\item $p_v\in single(N(v)) \land p_v=AskFirst(v) \land p_{p_{v}}=v$;
	\item $p_u\in single(N(u)) \land p_u=AskSecond(u) \land p_{p_{u}}=u$;
	\item $s_v = True$;
	\item $u$ does not execute moves between $C_1$ and $C_3$;
	\item $end_v = True$;
	\end{enumerate}
\end{enumerate}
\end{enumerate}
\end{lemma}

\begin{proof}
We show Point 1a. For $u$ to write $True$ in transition $C_0 \mapsto C_1$, $u$ executes a $MatchSecond$ in this transition. Thus $s_v=True$ must hold in $C_0$ and $p_v=AskFirst(v)$ as well.  By Corollary \ref{upd}, $u$ cannot write $True$ consecutively if $v$ does not execute any move. Thus at some point before $C_0$, $v$ applied some rule. Thus, and by Lemma \ref{ess},  $\exists x\in single(N(v)) : p_v =x \land p_x = v$ in configuration $C_0$, so Point 1b holds. 

As $AskFirst(v)\neq null$ in $C_0$, $v$ is not eligible for $ResetMatch$ in $C_0$. We prove now that $v$ is not eligible for $Update$. By Corollary \ref{upd} and Lemma \ref{lem:UpdateL1}, line 1 of the $Update$ guard does not hold in $C_0$. Finally, according to Point 2b, the second line of the $Update$ guard does not hold, which concludes Point 2.

We consider now Point 3a. In $C_1$, $s_u=end_u=True$ holds because, executing a $MatchSecond$, $u$ writes $True$ in $end_u$ and writes $end_u$ in $s_u$ during transition $C_0 \mapsto C_1$.

We now show Point 3b. $AskFirst(v)$ and $AskSecond(u)$ remain constant in $C_0\mapsto C_1$ since neither $u$ nor $v$ execute an $Update$ in this transition. Moreover, the only rule $v$ can execute in $C_0\mapsto C_1$ is a $MatchFirst$, according to Point 2. Thus $v$ does not change its $p$-value in $C_0\mapsto C_1$ and so $p_v=AskFirst(v)$ in $C_1$. Now, in $C_0$, $v\in matched(N(p_v)) \land p_{p_v}=v$ thus $p_v$ cannot execute $UpdateP$ in $C_0 \mapsto C_1$ and thus it cannot change its $p$-value. So, $p_{p_v}=v$ in $C_1$.

Point 3c holds since after $u$ executed a $MatchSecond$ in $C_0 \mapsto C_1$, observe that necessarily $p_u=AskSecond(u)$ in $C_1$. Moreover, $s_u = True$ in $C_1$ so, according to Lemma \ref{ess}, $\exists y\in single(N(u)) : p_u =y \land p_y = u$ in $C_1$.

$p_v =AskFirst(v)$ and $p_{p_{v}}=v$ hold in $C_0$, according to Points 2a and 2b. Moreover, $p_u=AskSecond(u)$ holds in $C_0$ since $u$ writes $True$ in $end_u$ while executing a $MatchSecond$ in \\
$C_0\mapsto C_1$. Finally, by Point 2, $v$ can only execute $MatchFirst$ in $C_0 \mapsto C_1$, thus variable $s_v$ remains $True$ in transition $C_0 \mapsto C_1$ and Point 3d holds.

We now prove Point 4. If $end_v=True$ in $C_1$, then according to Lemma \ref{lemm:stableEnd}, $u$ is not eligible for any rule in $C_1$. Now, let us consider the case $end_v=False$ in $C_1$. By Points 3c and 3d, $u$ is not eligible for $ResetMatch$. By Point 3c and Lemma \ref{lem:UpdateL1}, $u$ is not eligible for $Update$. By Points 3a, 3b and 3c, $u$ is not eligible for $MatchSecond$. Finally, since $u$ is Second in $C_1$, $u$ is not eligible for $MatchFirst$ neither and Point 4 holds. 

Now since between $C_1$ and $C_2$, $v$ does not execute any rule (by Point 5b), and since $p_u$ (resp. $p_v$) is not eligible for $UpdateP$ while $u$ (resp. $v$) does not move (because $p_{p_u}=u$ (resp. $p_{p_v}=v$)), then $Ask(u), Ask(v)$, $p_{p_u}$ and $p_{p_v}$ remain constant while $u$ does not make any move. And so, properties 3a, 3b, 3c and 3d hold for any configuration between $C_1$ and $C_2$, thus $u$ is not eligible for any rule between $C_1$ and $C_2$ and $u$ will not execute any move from $C_1$ to $C_3$. Moreover, the $end_v$-value is the same from $C_1$ to $C_2$. 

If $end_v=False$ in $C_2$, then  $v$ is eligible for a $MatchFirst$ and that it will write $True$ in its $end_v$-variable while all properties of Point 3 will still hold in $C_3$. Thus Point 5 holds. 
\end{proof}


\begin{theorem}\label{nber}
In any execution, a matched node $u$ can write  $end_{u} := True$ at most twice.
\end{theorem}

\begin{proof}
Let $(u,v)$ be a matched edge and $\Exe$ be an execution where $u$ writes $True$ in its variable $end_u$ at least twice. Let $C_0\mapsto C_1$ be the transition where $u$ writes $True$ in $end_u$ for the second time in $\Exe$. 
 If $u$ is First (resp. Second) in $C_0$ then from  Lemmas  \ref{lemm:stableEnd} and \ref{lemm:stableFirst}, (resp.~\ref{lemm:stableSecond}), from C1,  neither $u$ nor $v$ will ever be eligible for any rule.
 \end{proof}

\subsection{The number of times single nodes can change their $end$-variable}\label{section:part2}

Recall that $\matchNB$ is the number of matched nodes and $\singleNB$ is the number of single nodes.

\begin{lemma}\label{one}
Let $x$ be a single node. If $x$ writes $True$ in some transition $C_0 \mapsto C_1$ then, in $C_0$, $\exists u \in matched(N(x)):p_x = u \land p_u=x \land end_x=False \land end_u=True$.
\end{lemma}

\begin{proof}
To write True in its $end$ variable, a single node must apply $UpdateEnd$. Observe now that to apply this rule, the conditions described in the Lemma must hold.
\end{proof}


\begin{lemma}\label{pa}
Let $u$ be a matched node. Consider an execution $\Exe$ starting after $u$ executed some rule and in which $end_u$ is always $True$, except for the last configuration $D$ of $\Exe$ in which it may be $False$. Let $\Exe\backslash D$ be all configurations of $\Exe$ but configuration $D$. In $\Exe\backslash D$, the following holds:
	\begin{itemize}
	\itemshort
	\item $p_u  \in single(N(u))$;
	\item $p_u$ remains constant.
	\end{itemize}
\end{lemma}

\begin{proof}
Since $end_u=True$ in $\Exe\backslash D$, the last rule executed before $\Exe$ is necessarily a $Match$ rule. So, at the beginning of $\Exe$, $p_u\in single(N(u))$, otherwise, $u$ would not have executed a $Match$ rule, but an $Update$ instead. 

We prove now that in $\Exe\backslash D$, $p_u$ remains constant. 
Assume by contradiction that there exists a transition in which $p_u$ is modified. Let $C_0 \mapsto C_1$ be the first such transition. First, observe that in $\Exe\backslash D$, $u$ cannot execute $ResetMatch$ nor $Update$ since that would set $end_u$ to $False$. Thus $u$ must execute a $Match$ rule in $C_0 \mapsto C_1$. Since the value of $p_u$ changes in this transition, this implies that $Ask(u)\neq p_u$ in $C_0$. Thus, whatever the $Match$ rule, observe now that in $C_1$, $end_u$ must be $False$, which gives a contradiction and concludes the proof.
\end{proof}


\begin{definition}\label{copy}
Let $u$ be a matched node. We say that a transition $C_0\mapsto C_1$ is of type \emph{"a single copies True from $u$"} if it exists a single node $x$ such that  $(p_x=u \land p_u=x \land end_x=False)$ in $C_0$ and $end_x=True$ in $C_1$. Notice that by Lemma \ref{one}, $end_u=True$ in $C_0$ and $x \in single(N(u))$.

If a transition $C_0\mapsto C_1$ is of type \emph{"a single node copies True from $u$"} and if $x$ is the single node with $(p_x=u \land p_u=x \land end_x=False)$ in $C_0$ and $end_x=True$ in $C_1$, then we will say \emph{$x$ copies $True$ from $u$}.
\end{definition}


\begin{lemma}\label{theol}
Let $u$ be a matched node and $\Exe$ be an execution. In $\Exe$, there are at most three transitions of type \emph{"a single copies $True$ from $u$"}. 
\end{lemma}

\begin{proof}
Let $\Exe$ be an execution. We consider some sub-executions of $\Exe$. 

Let $\Exe_{init}$ be a sub-execution of $\Exe$ that starts in the initial configuration of $\Exe$ and that ends just after the first move of $u$.  Let $C_0 \mapsto C_1$ be the last transition of $\Exe_{init}$. Observe that $u$ does not execute any move until configuration $C_0$ and executes its first move in transition $C_0 \mapsto C_1$. We will write $\Exe_{init} \setminus C_1$ to denote all configurations of $\Exe_{init}$ but the configuration $C_1$. We prove that there is at most one transition of type \emph{"a single copies $True$ from $u$"} in $\Exe_{init}$.

There are two possible cases regarding $end_u$ in all configuration of $\Exe_{init} \setminus C_1$: either $end_u$ is always $True$ or $end_u$ is always $False$. 
If $end_u=False$ then by Definition \ref{copy}, no single node can copy $True$ from $u$ in $\Exe_{init}$, not even in transition $C_0\mapsto C_1$, since no single node is eligible for such a copy in $C_0$. 
If $end_u=True$, once again, there are two cases: either (i) $(p_u=null \lor p_u \notin single(N(u)))$ in all configuration of $\Exe_{init} \setminus C_1$, or (ii) $(p_u \in single(N(u)))$ in $\Exe_{init} \setminus C_1$. In case (i) then by Definition \ref{copy} no single node can copy $True$ from $u$ in $\Exe_{init}$, not even in $C_0\mapsto C_1$. In case (ii), observe that $p_u$ remains constant in all configurations of $\Exe_{init} \setminus C_1$, thus at most one single node can copy $True$ from $u$ in $\Exe_{init}$.

Let $\Exe_{true}$ be a sub-execution of $\Exe$ starting after $u$ executed some rule and such that: for all configurations in $\Exe_{true}$ but the last one, $end_u=True$.  There is no constraint on the value of $end_u$ in the last configuration of $\Exe_{true}$. 
According to Lemma \ref{pa}, $p_u \in single(N(u))$ and $p_u$ remains constant in all configurations of $\Exe_{true}$ but the last one. This implies that at most one single can copy $True$ from $u$ in $\Exe_{true}$.

Let $\Exe_{false}$ be an execution starting after $u$ executed some rule and such that: for all configurations in $\Exe_{false}$ but the last one, $end_u=False$. There is no constraint on the value of $end_u$ in the last configuration of $\Exe_{false}$. By Definition \ref{copy}, no single node will be able to copy $True$ from $u$ in $\Exe_{false}$.

To conclude, by Corollary \ref{nber}, $u$ can write $True$ in its $end$ variable at most twice. Thus, for all executions $\Exe$, $\Exe$ contains exactly one sub-execution of type $\Exe_{init}$, and at most two sub-executions of type $\Exe_{true}$ and the remaining sub-executions are of type $\Exe_{false}$. This implies that in total, we have at most three transitions of type "\emph{a single copies $True$ from $u$}" in $\Exe$.
\end{proof}


\begin{lemma}\label{kk}
In any execution, the number of transitions where a single node writes $True$ in its $end$ variable is at most $3\matchNB$.
\end{lemma}

\begin{proof}
Let $\Exe$ be an execution and $x$ be a single node. If $x$ writes $True$ in $end_x$ in some transition of $\Exe$, then $x$ necessarily executes an $UpdateEnd$ rule and by Definition \ref{copy}, this means \emph{$x$ copies $True$ from some matched node} in this transition.  Now the lemma holds by Lemma \ref{theol}.
\end{proof}


\begin{lemma}\label{big}
In any execution, the number of transitions where a single node changes the value of its $end$ variables (from $True$ to $False$ or from $False$ to $True$) is  at most $\singleNB+6 \matchNB$ times.
\end{lemma}

\begin{proof}
A single node can write $True$ in its $end$ variable at most $3 \matchNB$ times, by Corollary \ref{kk}. Each of this writing allows one writing from $True$ to $False$, which leads to $6 \matchNB$ possible modifications of the $end$ variables. Now, let us consider a single node $x$. If $end_x=False$ initially, then no more change is possible, however if $end_x=True$ initially, then one more modification from $True$ to $False$ is possible. Each single node can do at most one modification due to this initialization and thus the Lemma holds. 
\end{proof}

\subsection{How many $Update$ in an execution?}\label{section:part3}

\begin{definition}
Let $u$ be a matched node and $C$ be a configuration. We define $Cand(u,C)=\{  x\in single(N(u)) : (p_x = u ~\vee~ end_x = False)  \}$ which is the set of vertices considered by the function $BestRematch(u)$ in configuration $C$.
\end{definition}


\begin{lemma}\label{flip}
Let $u$ be a matched node that has already executed some rule. If there exists a transition $C_0 \mapsto C_1$ such that $u$ is eligible for $Update$ in $C_1$ and not in $C_0$, then there exists a single node $x$ such that $x \in Cand(u,C_0) \backslash Cand(u,C_1)$ or $x\in Cand(u,C_1) \backslash Cand(u,C_0)$. Moreover, in transition $C_0 \mapsto C_1$, $x$ flips the value of its end variable.
\end{lemma}

\begin{proof}
Since $u$ has already executed some rule, to become eligible for $Update$ in transition $C_0 \mapsto C_1$, necessarily the second disjonction in the $Update$ rule must hold, by Lemma~\ref{lem:UpdateL1}. This implies that $(\alpha_v, \beta_v) \ne BestRematch(v)$ must become $True$ in $C_0 \mapsto C_1$. Now either $Lowest(Cand(u,C_0))\notin Cand(u,C_1)$ or $\exists x \notin Cand(u,C_0)$ such that $x= Lowest(Cand(u,C_1))$. This proves the first point.

For the second point we first consider the case $x\in Cand(u,C_1)$ and $x\notin Cand(u,C_0)$. Necessarily $end_x=True\land p_x \neq u $ in $C_0$ and $end_x=False \lor p_x = u $ in $C_1$. If $p_x=u$ in $C_1$ then in transition $C_0 \mapsto C_1$, $x$ has executed an $UpdateP$ and the second point holds. Assume now that $p_x\neq u$ in $C_1$. Necessarily $end_u=False$ in $C_1$ and the Lemma holds. 
 
We consider the second case in which $x\notin Cand(u,C_1)$ and $x\in Cand(u,C_0)$. Necessarily in $C_1$, $p_x \neq v$ and $ end_x=True$. Thus if $end_x=False$ in $C_0$ the lemma holds.  Assume by contradiction that $end_x=True$ in $C_0$.  This implies $p_x=u$ in $C_0$. But since in $C_1$ $p_x \neq u$ then $x$ executed either $UpdateP$ or $UpdateEnd$ in $C_0 \mapsto C_1$ which implies $end_x=False$ in $C_1$, a contradiction. This completes the proof.
\end{proof}


\begin{corollary}\label{update}
Matched nodes can execute at most $\Delta(\singleNB+6 \matchNB) + \matchNB$ times the $Update$ rule.
\end{corollary}

\begin{proof}
Initially each matched node can be eligible for an $Update$. Now, let us consider only matched nodes  that have  already executed a move. For such a node to become eligible for an $Update$ rule, at least one single node must change the value of its $end$ variable by Lemma \ref{flip}. Thus, each change of the $end$ value of a single node can generate at most $\Delta$ matched nodes to be eligible for an $Update$. By Lemma \ref{big}, the number of transitions where a single node changes the value of its $end$ variables is  at most $\singleNB+6 \matchNB$ times. Thus we obtain at most $\Delta(\singleNB+6 \matchNB)$ $Update$ generated by a change of the $end$ value of a single node and the Lemma holds.
\end{proof}

\subsection{A bound on the total number of moves in any execution}\label{section:part4}

\begin{definition}\label{chocapic}
In the following, we call $\noUpdate$, a finite execution that does not contain any executed $Update$ rule. Let $\Cstart$ be the first configuration of $\noUpdate$ and $\Cend$ be the last one.
\end{definition}

Observe that in the execution $\noUpdate$, all variables $\alpha$ and $\beta$ remain contant and thus, predicates $AskFirst$ and $AskSecond$ for all matched nodes remains constant too.


\begin{lemma}\label{null}
Let $(u,v)$ be a matched edge. If $Ask(u)=Ask(v)=null$ in $\noUpdate$, then $u$ and $v$ can both execute at most one $ResetMatch$.
\end{lemma}

\begin{proof}
Recall that in the execution $\noUpdate$, by definition, $u$ and $v$ do not execute the $Update$ rule. Moreover, these two nodes are not eligible for $Match$ rules since $Ask(u)=Ask(v)=null$. Thus they are only eligible for $ResetMatch$. Observe now it is not possible to execute tis rule twice in a row, which completes the proof.
\end{proof}


\begin{lemma}\label{debv}
Let $(u,v)$ be a matched edge. Assume that in $\noUpdate$, $u$ is $First$ and $v$ is $Second$. If $s_u$ is False in all configurations of $\noUpdate$ but the last one, then $v$ can execute at most one rule in $\noUpdate$.
\end{lemma}

\begin{proof}
Since $s_u=False$ in all configurations of $\noUpdate$ but the last one, node $v$ which is $Second$ can only be eligible for $ResetMatch$. Observe that if $v$ executes $ResetMatch$, it is not eligible for a rule anymore and the Lemma holds.  
\end{proof}


\begin{lemma}\label{debu}
Let $(u,v)$ be a matched edge. Assume that in $\noUpdate$, $u$ is $First$ and $v$ is $Second$. If $s_u$ is False throughout  $\noUpdate$, then $u$ can execute at most one rule in $\noUpdate$.
\end{lemma}

\begin{proof}
Node $u$ can only be eligible for $MatchFirst$. Assume $u$ executes $MatchFirst$ for the first time in some transition $C_0 \mapsto C_1$, then in $C_1$, necessarily, $p_u=AskFirst(u)$, $s_u=False$ (by hypothesis) and $end_u=False$ by Lemma \ref{Trr}. Let $\noUpdate_1$ be the execution starting in $C_1$ and finishing in $\Cend$. Since in $\noUpdate_1$, there is no $Update$, observe that $p_u=AskFirst(u)$ remains $True$ in this execution. Assume by contradiction that $u$ executes another $MatchFirst$ in $\noUpdate_1$. Consider the first transition $C_2 \mapsto C_3$ after $C_1$ when it executes this rule. Notice that between $C_1$ and $C_2$ it does not execute rules. Thus in $C_2$, $p_u=AskFirst(u)$, $s_u=False$ and $end_u=False$ hold. Now if $u$ executes $MatchFirst$ in $C_2$ it is necessarily to modify the value of $s_u$ or $end_u$. By definition, it cannot change the value of $s_u$. Moreover it cannot modify the value of $end_u$ as this would imply by Lemma \ref{Trr} that $s_u=True$ in $C_3$. This completes the proof.
\end{proof}


\begin{lemma}\label{fin}
Let $(u,v)$ be a matched edge. Assume that in $\noUpdate$, $u$ is $First$, $v$ is $Second$ and that $u$ writes $True$ in $s_u$ in some transition of $\noUpdate$. 
Let $C_0 \mapsto C_1$ be the transition in $\noUpdate$ in which $u$ writes $True$ in $s_u$ for the first time. Let $\noUpdate_1$ be the execution starting in $C_1$ and finishing in $\Cend$. In $\noUpdate_1$, $u$ can apply at most 3 rules and $v$ at most 2.
\end{lemma}

\begin{proof}
We first prove that in $\noUpdate_1$, $s_u$ remains $True$. 
Observe that $u$ cannot execute $Update$ neither $ResetMatch$ since it is $First$. So $u$ can only execute $MatchFirst$ in $\noUpdate_1$. For $u$ to write $False$ in $s_u$, it must exists a configuration in $\noUpdate_1$ such that $p_u\neq AskFirst(u)\lor p_{p_u}\neq u \lor p_v\not \in \{AskSecond(v), null\}$. Let us prove that none of these cases are possible. 

Since $u$ executed $MatchFirst$ in transition $C_0 \mapsto C_1$ writting $True$ in $s_u$ then, by definition of this rule, $p_u= AskFirst(u)\land p_{p_u}= u \land p_v\in \{AskSecond(v), null\}$ holds in $C_0$.  
As there is no $Update$ in $\noUpdate$, then $AskFirst(u)$ and $AskSecond(v)$ remains constant throughout $\noUpdate$ (and $\noUpdate_1$). So each time $u$ executes a $MatchFirst$, it writes the same value $AskFirst(u)$ in its $p$ value. Thus $p_u= AskFirst(u)$ holds throughout $\noUpdate_1$. Moreover, each time $v$ executes a rule, it writes either $null$ or the same value $AskSecond(v)$ in its $p$ value. Thus $p_v\in \{AskSecond(v), null\}$ holds throughout $\noUpdate_1$.
Now by Lemma \ref{ess}, in $C_1$ we have, $\exists x\in single(N(u)) : p_u =x \land p_x = u$, since $s_u=True$ . This stays $True$ in $\noUpdate_1$ as $p_u$ remains constant and $x$ will then not be eligible for $UpdateP$ in $\noUpdate_1$. Thus $p_{p_u}= u$ holds throughout $\noUpdate_1$. 
Thus, $p_u= AskFirst(u)\land p_{p_u}= u \land p_v\in \{AskSecond(v), null\}$ holds throughout $\noUpdate_1$ and so $s_u=True$ throughout  $\noUpdate_1$.

This implies that in $\noUpdate_1$, $v$ is only eligible for $MatchSecond$. The first time it executes this rule in some transition $B_0\mapsto B_1$, with $B_1\geq C_1$,  then in $B_1$, $p_v=AskSecond(v)$, $s_v=end_v$ and this will hold between $B_1$ and $\Cend$. If $end_v=True$ in $B_1$ then this will stay $True$ between $B_1$ and $\Cend$. Indeed, $p_v$ is not eligible for $UpdateP$ and we already showed that $p_u=AskFirst(u)$ holds in $\noUpdate_1$. In that case, between $B_1$ and $\Cend$, $v$ will not be eligible for any rule and so $v$ will have executed at most one rule in $\noUpdate_1$. In the other case, that is $end_v(=s_v)=False$ in $B_1$, since $p_v=AskSecond(v)$ holds between $B_1$ and $\Cend$, necessarily, the next time $v$ executes a $MatchSecond$ rule, it is to write $True$ in $end_v$. After that observe that $v$ is not eligible for any rule. Thus, $v$ can execute at most 2 rules in $\noUpdate_1$.

To conclude the proof it remains to count the number of moves of $u$ in $\noUpdate_1$. Recall that we proved that $s_u$ is always $True$ in $\noUpdate_1$. Thus whenever $u$ executes a $MatchFirst$, it is to modify the value of its $end$ variable. Observe that this value depends in fact of the value of $end_v$ and of $p_v$ since we proved $p_u= AskFirst(u)\land p_{p_u}= u \land s_u \land p_v\in \{AskSecond(v), null\}$ holds throughout $\noUpdate_1$. Since we proved that in $\noUpdate_1$, $v$ can execute at most two rules, this implies that these variables can have at most three different values in $\noUpdate_1$. Thus $u$ can execute at most 3 rules in $\noUpdate_1$.  
\end{proof}


\begin{lemma}\label{initi}
Let $(u,v)$ be a matched edge. Assume that in $\noUpdate$, $u$ is $First$ and $v$ is $Second$. If $s_u$ is $True$ throughout $\noUpdate$ and if $u$ does not execute any move in $\noUpdate$, then $v$ can execute at most two rules in $\noUpdate$.
\end{lemma}

\begin{proof}
By Definition \ref{chocapic}, $v$ cannot execute $Update$ in $\noUpdate_1$. Since we suppose that in $\noUpdate_1$, $s_u=True$ then $v$ is not eligible for $ResetMatch$. Thus in $\noUpdate_1$, $v$ can only execute $MatchSecond$. After it executed this rule for the first time, $p_v=AskSecond(v)$ and $s_v=end_v$ will always hold, since $v$ is only eligible for $MatchSecond$. Thus the second time it executes this rule, it is necessarily to modify its $end_v$ and $s_v$ variables. Observe that after that, since $u$ does not execute rules,  $v$ is not eligible for any rule.
\end{proof}


\begin{lemma}\label{bupd}
Let $(u,v)$ be a matched edge. In $\noUpdate$, $u$ and $v$ can globally execute at most $12$ rules.
\end{lemma}

\begin{proof}
If $Ask(u)=Ask(v)=null$, the Lemma holds by Lemma \ref{null}. Assume now that $u$ is $First$ and $v$ $Second$. We consider two executions in $\noUpdate$. 

Let $C_0 \mapsto C_1$ be the first transition in $\noUpdate$ in which $u$ executes a rule. Let $\noUpdate_0$ be the execution starting in $\Cstart$ and finishing in $C_0$. There are two cases.

If $s_u=False$ in  $\noUpdate_0$ then $v$ is only eligible for $ResetMatch$ in this execution. Observe that after it executes this rule for the first time in $\noUpdate_0$, it is not eligible for any rule after that in $\noUpdate_0$.

If $s_u=True$ in  $\noUpdate_0$ then by Lemma \ref{initi}, $v$ can execute at most two rules in this execution. In transition $C_0 \mapsto C_1$, $u$ and $v$ can execute one rule each.

Let $\noUpdate_1$ be the execution starting in $C_1$ and finishing in $\Cend$. Whatever rule $u$ executes in transition $C_0 \mapsto C_1$ observe that $u$ either writes $True$ or $False$ in $s_u$. If $u$ writes $True$ in $s_u$ in transition $C_0 \mapsto C_1$, then by Lemma \ref{fin}, $u$ and $v$ can execute at most five rules in total in $\noUpdate_1$.

Consider the other case in which $u$ writes $False$ in $C_1$. Let $C_2\mapsto C_3$ be the first transition in $\noUpdate_1$ in which $u$ writes $True$ in $s_u$. Call $\noUpdate_{10}$ the execution between $C_1$ and $C_3$ and $\noUpdate_{11}$ the execution between $C_3$ and $\Cend$. By definition, $s_u$ stays $False$ in $\noUpdate_{10} \backslash C_3$. Thus in $\noUpdate_{10} \backslash C_3$, $u$ can execute at most one rule, by {Lemma} \ref{debu}. Now in $\noUpdate_{10}$, $u$ can execute at most two rules. By Lemma \ref{debv}, $v$ can execute at most one rule in $\noUpdate_{10}$. In total, $u$ and $v$ can execute at most three rules in $\noUpdate_{10}$. In $\noUpdate_{11}$, $u$ and $v$ can execute at most five rules by Lemma \ref{fin}. Thus in $\noUpdate_1$, $u$ and $v$ can apply at most eight rules.
\end{proof}


\begin{theorem}\label{matched}
In any execution, matched nodes can execute at most $12\matchNB (\Delta(\singleNB+6 \matchNB) + \matchNB)$ rules.
\end{theorem}

\begin{proof}
By Lemma \ref{update}, matched nodes can execute at most $\Delta(\singleNB+6 \matchNB) + \matchNB$ times the $Update$ rule. By Lemma \ref{bupd}, between the execution of two $Updates$ rules, a matched node can execute at most $12$ rules, which concludes the proof.
\end{proof}


\begin{lemma}\label{lem:ResetEndMove}
In any execution, single nodes can execute at most $\singleNB$ times the $ResetEnd$ rule.
\end{lemma}

\begin{proof}
We  prove that a single node $x$ can execute the $ResetEnd$ rule at most once. Assume by contradiction that it executes this rule twice. Let $C_0\mapsto C_1$ be the transition when it executes it the second time. In $C_0$, $end_x=True$, by definition of the rule. Since $x$ already executed a $ResetEnd$ rule, it must have some point wrote $True$ in $end_x$. This is only possible through an execution of $UpdateEnd$. Thus consider the last transition $D_0\mapsto D_1$ in which it executed this rule. Observe that $D_1 \leq C_0$. Since between $D_1$ and $C_0$, $end_x$ remains $True$, observe that $x$ does not execute any rule between these two configurations. Now since in $D_1$, $p_x\neq null$ and this holds in $C_0$ then $x$ is not eligible for $ResetEnd$ in $C_0$, which gives the contradiction. This implies that single nodes can execute at most $\mathcal{O}(\singleNB)$ times the $ResetEnd$ rule. 
\end{proof}


\begin{lemma}\label{lem:updateEndMove}
In any execution, single nodes can execute at most $\singleNB+6 \matchNB$ times the $UpdateEnd$ rule.
\end{lemma}

\begin{proof}
By Lemma \ref{big}, single nodes can change the value of their $end$ variable at most $\singleNB+6 \matchNB$ times. Thus they can apply $UpdateEnd$ at most $\singleNB+6 \matchNB$  times, since in every application of this rule, the value of the $end$ variable must change.
\end{proof}


\begin{lemma}\label{lem:updateMove}
In any execution, single nodes can execute at most $12\matchNB \times (\Delta(\singleNB+6 \matchNB) + \matchNB)+1$ times the $UpdateP$ rule.
\end{lemma}

\begin{proof}
Let $x$ be a single node. 
Let $C_0 \mapsto C_1$ be a transition in which $x$ executes an $UpdateP$ rule and let $C_2 \mapsto C_3$ be the next transition after $C_1$ in which $x$ executes an $UpdateP$ rule. We prove that for $x$ to execute the $UpdateP$ rule in $C_2\mapsto C_3$, a matched node had to execute a move between $C_0$ and $C_2$. 

In $C_1$ there are two cases: either $p_x=null$ or $p_x\neq null$. Assume to begin that $p_x=null$. This implies that in $C_0$ the set $\{w \in N(x) | p_w =x \}$ is empty. In $C_2$, $p_x=null$, since between $C_1$ and $C_2$, $x$ can only apply $UpdateEnd$ or $ResetEnd$. Thus if it applies $UpdateP$ in $C_2$, necessarily $\{w \in N(x) | p_w =x \}\neq \emptyset$. This implies that a matched node must have executed a $Match$ rule between $C_1$ and $C_2$ and the lemma holds in that case.

Consider now the case in which $p_x=u$ with $u\neq null$ in $C_1$. By definition of the $UpdateP$ rule, we also have $u \in matched(N(x)) \land p_u=x$ holds in $C_0$. In $C_2$ we still have that $p_x=u$ since between $C_1$ and $C_2$, $x$ can only execute $UpdateEnd$ or $ResetEnd$. Thus if $x$ executes $UpdateP$ in $C_2$, necessarily $p_{p_{x}}\neq x$. This implies that $p_u\neq x$ and so $u$ executed a rule between $C_0$ and $C_2$. 

Now, the lemma holds by Theorem \ref{matched}.
\end{proof}


\begin{corollary}
In any execution, nodes can execute at most $\mathcal{O}(n^3)$ moves.
\end{corollary}

\begin{proof}
According to Lemmas \ref{lem:ResetEndMove}, \ref{lem:updateEndMove} and \ref{lem:updateMove}, single nodes can execute at most $\mathcal{O}(n^3)$ moves.
Moreover, according to Theorem \ref{matched}, matched nodes can execute at most $\mathcal{O}(n^3)$ moves.
\end{proof}


\begin{corollary}\label{cor:end}
The algorithms $\algoName$ converges in $O(n^3)$ steps under the adversarial distributed daemon. 
\end{corollary}

\bibliographystyle{plainurl}
\bibliography{matching}

\end{document}